\documentclass[runningheads,a4paper]{llncs}

\usepackage{amssymb}
\usepackage{graphicx}
\usepackage{subcaption}
\usepackage{verbatim}

\usepackage{url}
\urldef{\mailsa}\path|{maanders, meierflo, patrick.schnider, steger}@inf.ethz.ch|    
\newcommand{\keywords}[1]{\par\addvspace\baselineskip
\noindent\keywordname\enspace\ignorespaces#1}

\usepackage[utf8]{inputenc} 
\usepackage[T1]{fontenc}    
\usepackage{lmodern}
\usepackage{hyperref}       
\usepackage{url}            

\usepackage{tikz}
\usepackage{amsmath}

\newcommand{\N}{\ensuremath{\mathbb{N}}}

\newcommand{\abs}[1]{\ensuremath{\left| #1 \right|}}

\newcommand{\perm}{P}

\begin{document}

\mainmatter  

\title{Even flying cops should think ahead}

\titlerunning{Even flying cops should think ahead}

%
%
\author{Anders Martinsson
\and Florian Meier \and Patrick Schnider \and Angelika Steger }
%

\institute{Department of Computer Science\\
  ETH Zürich, Switzerland\\
  \mailsa\\
}

%
%

\toctitle{Lecture Notes in Computer Science}
\tocauthor{Authors' Instructions}
\maketitle

\begin{abstract}
We study the entanglement game, which is a version of cops and robbers, on sparse graphs. While the minimum degree of a graph $G$ is a lower bound for the number of cops needed to catch a robber in $G$, we show that the required number of cops can be much larger, even for graphs with small maximum degree.
In particular, we show that there are 3-regular graphs where a linear number of cops are needed.

\end{abstract}

\section{Introduction}\label{sec:Introduction}
The game of cops and robbers was first introduced and popularised in the 1980's by Aigner and Fromme~\cite{aigner84}, Nowakowski and Winkler~\cite{nowakowski83} and Quilliot~\cite{quilliot78}.
Since then, many variants of the game have been studied, for example where cops can catch robbers from larger distances~(\cite{bonato10}), the robber is allowed to move at different speeds~(\cite{alon15,frieze12}), or the cops are lazy, meaning that in each turn only one cop can move~(\cite{bal13,bal16}). 

In this paper we consider 
the \emph{entanglement game}, introduced by Berwanger and Gr{\"a}del~\cite{berwanger04} that is the following version of the cops and robbers game on a (directed or undirected) graph $G$.
First, the robber chooses a starting position and the $k$ cops are outside the graph. In every turn, the cops can either stay where they are, or they can fly one of them to the current position of the robber.
Regardless of whether the cops stayed or one of them flew to the location of the robber, the robber then has to move to a neighbor of his current position that is not occupied by a cop.
If there is no such neighbor, the cops win. The robber wins if he can run away from the cops indefinitely.
The \emph{entanglement number} of a graph $G$, denoted by $\text{ent}(G)$, is the minimal integer $k$ such that $k$ cops can catch a robber on $G$.
In order to get accustomed to the rules of the game, it is a nice exercise to show that the entanglement number of an (undirected) tree is at most 2.

The main property that distinguishes the entanglement game from other variants of cops and robbers is the restriction that the cops are only allowed to fly to the current position of the robber.
This prevents the cops from cutting off escape routes or forcing the robber to move into a certain direction.
As we will show in this paper, it is this restriction that enables the robber to run away from many cops.

In a similar way to how the classical game of cops and robbers can be used to describe the treewidth of a graph, the entanglement number is a measure of how strongly the cycles of the graph are intertwined, see~\cite{berwanger12}.
Just like many problems can be solved efficiently on graphs of bounded treewidth, Berwanger and Gr{\"a}del~\cite{berwanger04} have shown that parity games of bounded entanglement can be solved in polynomial time.

As the cops do not have to adhere to the edges of the graph $G$ in their movement, adding more edges to $G$ can only help the robber.
In fact, it can be seen easily that on the complete graph $K_n$, $n-1$ cops are needed to catch the robber.
Furthermore, observe that the  minimum degree of the graph $G$ is a lower bound on the entanglement number, as otherwise the robber will always find a free neighbor to move to.
These observations seem to suggest that, on sparse graphs, the cops should have an advantage and therefore few cops would suffice to catch the robber. Indeed, on 2-regular graphs, it is easily checked that three cops can always catch the robber.

Motivated by this, we study the entanglement game on several classes of sparse graphs. We show that for sparse Erd\H{o}s-R\'enyi random graphs, with high probability linearly many cops are needed. We then apply similar ideas to show our main result, c.f.~Theorem~\ref{thm:random_matchings}, which states that the union of three random perfect matchings is with high probability such that the robber can run away from $\alpha n$ cops, for some constant $\alpha>0$. Further, we show in Theorem \ref{thm:upper_bound} that for any $3$-regular graph $\lfloor\frac{n}{4}\rfloor+4$ cops suffice to catch the robber. Finally,  we consider the entanglement game for the family of graphs given by a Hamilton cycle and a perfect matching connecting every vertex to its diagonally opposite vertex. 
For these graphs it may seem that the diagonal ``escape'' edges are quite nice for the robber. This, however, is not so: we show in Proposition~\ref{prop:DGn} that for these graphs six cops are always sufficient. However, we also show, cf.~Corollary~\ref{thm:random_matchings}, that if we replace this specific perfect matching by a random one, then with high probability a linear number of cops is needed.
Thus, in contrast to the intuition that sparse graphs are advantageous for the cops, they are often not able to use the sparsity to their advantage.

As mentioned previously, the entanglement game and entanglement number is also defined for directed graphs, the difference being that the robber moves to a successor of his current position.
In fact, motivated by an application in logic~(\cite{berwanger04logic}), the authors of the original papers about the entanglement game~(\cite{berwanger04,berwanger12}) focused on the directed version.
As a corollary of our main result we construct directed graphs of maximum (total) degree 3 which again require linearly many cops.
Our result on Erd\H{o}s-R\'enyi graphs also easily carries over to directed graphs.

\section{Results} \label{sec:Results}
In this section, we present and motivate our results.
In order to increase the readability, we postpone the proofs to Section \ref{sec:proof_main_thm}.
We start by analyzing the entanglement game on sparse random graphs.

\begin{theorem}\label{thm:Gnp}
For every $0<\alpha < 1$ there exists a constant $C=C(\alpha) >0$ such that for any $p\ge C/n$ the random graph $G_{n,p}$ is with high probability such $\alpha n$ cops do not suffice. The same result holds for directed random graphs.
\end{theorem}

Note that that $G_{n,p}$ with $p= C/n$ has constant average degree.
On the other hand, it is known that the maximum degree is $\log(n)$.
In the following, we construct graphs that need linearly many cops and have maximum degree $3$. 
The idea is that we define $G$ as the union of three random perfect matchings. Extending the proof ideas from Theorem~\ref{thm:Gnp}, we obtain the following results. 

\begin{theorem}\label{thm:random_matchings}
There exists an $\alpha >0$ such that with high probability $\alpha n$ cops do not suffice to catch the robber on the graph $G=M_1 \cup M_2 \cup M_3$, where $M_1, M_2, M_3$ are independent uniformly chosen random perfect matchings.
\end{theorem}



\begin{corollary}\label{cor:3-regular}
There exists an $\alpha >0$ and an $n_0\in\N$ such that for every even $n\ge n_0$ there exists a $3$-regular graph on $n$ vertices for which $\alpha n$ cops are not enough to catch the robber. 
\end{corollary}

This result may look surprising at first sight. In particular, consider the $3$-connected $3$-regular graph $DG_n$ obtained by taking a Hamilton cycle of length $2n$ and connecting every vertex to its antipode by an edge.



\begin{proposition}\label{prop:DGn}
For the graph $DG_n$ 
six cops suffice to catch the robber.
\end{proposition}

 The fact that all diagonals go to a vertex that is ``furthest away'' may seem to make catching the robber quite difficult for the cops. However, as it turns out, the symmetry of the construction is the reason for the small entanglement number. If we replace the matching of diagonals by a {\em random} matching then the entanglement number is typically large again.

\begin{corollary}\label{cor:cycle_matching}
Consider the graph $G= H \cup M$ where $H$ is a Hamiltonian cycle and $M$ is a random perfect matching. There exists an $\alpha >0$ such that with high probability $\alpha n $ cops do not suffice to catch the robber on $G$. 
\end{corollary}

We complement these lower bounds by the following upper bounds.

\begin{theorem}\label{thm:upper_bound}
For any $3$-regular graph on $n$ vertices, $\lfloor\frac{n}{4}\rfloor+4$ 
cops suffice.
\end{theorem}

We now turn to directed graphs. Theorem \ref{thm:random_matchings} immediately implies that there are graphs that are the union of six perfect matchings on which a linear number of cops is needed: simply direct each edge in both directions.
However, there also exist directed  graphs with maximum out-degree two  and (total) maximum degree three can be very hard for the cops.

\begin{corollary}\label{cor:directed_blow_up}
 There exists an $\alpha >0$ and an $n_0\in\N$ such that for every even $n\ge n_0$  there is a directed $3$-regular (that is, the sum of in and out degree of every vertex is $3$) graph $G$ on $6n$ vertices, such that $\alpha n$ cops do not suffice to catch the robber on $G$. 
\end{corollary}

The idea here is that we ``blow up'' an undirected $3$-regular graph to a directed one by replacing each vertex by a directed cycle of length six.
We note that, in contrast to the undirected version, we cannot just take a union of three random directed matchings. This follows from the fact that the largest strongly connected component in the union of three random directed matchings contains with high probability only sublinearly many vertices  as can be shown using some ideas from \cite{cooper2004size}. 

\section{Proofs} \label{sec:proof_main_thm}
We start the proof section by considering the entanglement game for Erd\H{o}s-R\'enyi random graphs $G_{n,p}$, cf. \cite{Bol,JLR} or \cite{FK} for an introduction to random graphs.  In Section~\ref{sec:proof_3_regular} we will  generalize this proof to obtain Theorem~\ref{thm:random_matchings}. In Section \ref{sec:proof_of_corollaries}, we use Theorem \ref{thm:random_matchings} to prove Corollaries~\ref{cor:3-regular}, \ref{cor:cycle_matching} and \ref{cor:directed_blow_up}. Finally, in Section \ref{sec:proofs_upper_bounds}, we will prove the stated upper bounds of Proposition \ref{prop:DGn} and Theorem \ref{thm:upper_bound}.


\subsection{Proof of Theorem~\ref{thm:Gnp}}\label{sec_proof_Gnp}
Recall from the introduction that adding edges to the graph can only make it harder for the cops. Without loss of generality it thus suffices to consider the case $G_{n,p}$, where $p=C/n$ for some (large) constant~$C>0$. A standard result from random graph theory is that such  a random graph has with high probability one large component  (of size approximately $\beta n$, where $\beta$ is a function of the constant~$C$ of the edge probability~$p$), while all additional components have at most logarithmic size. Here we need the following strengthening of this result:

\begin{lemma}\label{l:rg}
For every $0<\bar{\alpha} < 1$ there exists a constant $C=C(\bar{\alpha}) >0$ such that for any $p\ge C/n$ the random graph $G_{n,p}$ is with high probability such that every subset $X\subseteq V$ of size $\bar{\alpha} n$ induces a subgraph that
 contains a connected component of size at least $\frac23\bar{\alpha} n$.
\end{lemma} 
\begin{proof}
First observe that any graph on $\bar{\alpha} n$ vertices that does not contain a component of size $\beta n$ contains a cut $(S,\bar{S})$ such that $E(S,\bar S)$ contains no edge and $|S|,|\bar{S}|\le \frac12(\bar{\alpha}+\beta)n$. This follows easily by greedily placing components into $S$ as long as the size constraint is not violated. Note that such a cut $(S,\bar{S})$ contains at least $\frac14(\bar{\alpha}^2-\beta^2)n^2$ possible edges. The probability that all these edges are missing in the random graph $G_{n,p}$ is thus at most $(1-p)^{\frac14(\bar{\alpha}^2-\beta^2)n^2}\le e^{-\frac{C}4(\bar{\alpha}^2-\beta^2)n}$. By a union bound over all sets $X$ of size $\bar{\alpha} n$ and all possible cuts $(S,\bar{S})$ we thus obtain that the probability that the random graph $G_{n,p}$ does not satisfy the desired property is at most
$$
 \binom{n}{\bar{\alpha} n}\cdot 2^{\bar{\alpha} n} \cdot e^{-\frac{C}4(\bar{\alpha}^2-\beta^2)n},
$$ 
which in turn can be bounded by
$$
2^{H(\bar{\alpha})n} \cdot 2^{\bar{\alpha} n} \cdot e^{-\frac{C}4(\bar{\alpha}^2-\beta^2)n}
$$
by using the standard estimations for the binomial coefficient, where $H(x)=-x \log_2 x - (1-x)\log_2(1-x)$.
By letting $\beta = \frac23\bar{\alpha}$ and making $C$ large enough we see that this term goes to zero, which concludes the proof of the lemma.\qed
\end{proof}

With Lemma~\ref{l:rg} at hand we can now conclude the proof of the theorem as follows. Pick any $0<\alpha<1$ and let $\bar{\alpha}=1-\alpha$. Assume that $G_{n,p}$ satisfies the property of Lemma~\ref{l:rg}. The robber can win against $\alpha n$ cops with the following strategy: he aims at  always staying in a component of size at least $\frac23\bar\alpha n$ in a subgraph $G[A]$ for some cop-free set $A$ of size $\bar\alpha n$. Clearly, this can be achieved at the start of the game. Now assume that the robber is in such a component $C$ of a subgraph $G[A]$. If a cop moves to the  location of the robber,
we change $A$ by removing this vertex and  add instead another vertex not covered by a cop. Call this new cop-free set $A'$. By assumption $G[A']$ contains a component $C'$ of size $\frac23\bar\alpha n$. Clearly, $C$ and $C'$ overlap and the robber can thus move to $C'$, as required.

The directed cases follows similarly. The only slightly more tricky case is the argument for the existence of the cut $(S,\bar{S})$. This can be done as follows. Consider all strongly connected components of $G_{n,p}$. It is then not true that there exist no edges between these components. What is true, however, is that the cluster graph (in which the components are replaced by vertices and the edges between components by one or two directed edges depending on which type of edges exist between the corresponding components)  is acyclic. If we thus repeatedly consider sink components, placing them into  $S$ as long as the size constraint is not violated, we obtain a cut $(S,\bar{S})$ which does not contain any edge directed from $S$ to $\bar{S}$. From here on the proof is completed as before. \qed

\subsection{Proof of Theorem~\ref{thm:random_matchings} }\label{sec:proof_3_regular}
\def\perm{\mbox{pm}}%

%
In this section we prove the a slightly stronger version of Theorem \ref{thm:random_matchings}, i.e., we show that the statement of Theorem \ref{thm:random_matchings} holds with exponentially high probability. We need this statement to proof the corollaries in Section \ref{sec:proof_of_corollaries}. For completeness we restate the theorem in this form:

\begin{theorem}\label{thm:random_matchings_exp_high_prob}
There exists an $\alpha >0$ such that, with probability $1-e^{-\Omega(n)}$, $\alpha n$ cops do not suffice to catch the robber on the graph $G=M_1 \cup M_2 \cup M_3$, where $M_1, M_2, M_3$ are independent uniformly chosen random perfect matchings.
\end{theorem}

Let $\alpha>0$ denote a sufficiently small constant to be chosen later, and, as before, let $\bar{\alpha}:=1-\alpha$. The main idea of the proof of Theorem~\ref{thm:random_matchings_exp_high_prob} is similar to the strategy that we used in the proof of Theorem~\ref{thm:Gnp}. Namely, we show that every subgraph induced by an $\bar{\alpha}$-fraction of the vertices contains a large connected component.  In the proof of Theorem~\ref{thm:Gnp} we used $\frac{2}{3}\bar{\alpha} n$ as a synonym for ``large''. As it turns out, in the proof of Theorem~\ref{thm:random_matchings_exp_high_prob} we have to be more careful.
To make this precise we start with a definition.
Given $\alpha>0$, we say that a graph $G=(V,E)$ is $\alpha$-robust, if for every set $X\subseteq V$ of size $|X| \ge \bar{\alpha} n=(1-\alpha)n$ the induced graph $G[X]$ contains a connected component that is larger than $|X| / 2$.

\begin{lemma}\label{l:alpha:robust}
Assume $G=(V,E)$ is an $\alpha$-robust graph for some $0<\alpha<1$. Then  $\alpha|V|-2$ cops are not sufficient to catch the robber.
\end{lemma}
\begin{proof}

Let $n=|V|$ and 
assume there are at most $\alpha n-2$ cops.
The robber can win with the following strategy, similar to the one used on the random graphs: he aims at always staying in the unique component of size greater than $|A|/2$ in a subgraph $G[A]$ for some cop-free set $A$ of size $|A|\ge \bar{\alpha} n$ such that $|A|$ is even. In the beginning of the game, this can easily be achieved. If some cop is placed on the current position of the robber, then we remove this vertex from $A$ and add some other cop-free vertex arbitrarily to obtain a set $A'$. Let $C$ resp.\ $C'$ be the vertex set of the largest component in $G[A]$ resp.\ $G[A']$. It remains to show that $C$ and $C'$ overlap. Assume they do not. Then $|C\cup C'| = |C|+|C'|$. By $\alpha$-robustness and the evenness assumption of $A$ and $A'$ we have
$|C| \ge |A|/2 + 1$ and $|C'| \ge |A'|/2+1$. As $A'$ (and thus $C'$) contains at most one vertex that is not in $A$, this is a contradiction.\qed
\end{proof}

For the proof that the union of three perfect matchings is $\alpha$-robust for sufficiently small $\alpha$, we will proceed by contradiction. Here the following lemma will come in handy. 

\begin{lemma}\label{lemma:3part}
Let $G=(V,E)$ be a graph on $n=|V|$ vertices that does not contain a component on more than $n/2$ vertices.  Then  there exists a partition $V=B_1\cup B_2 \cup B_3$ such that
\begin{list}{}{}
\item[$(i)$\hfill]$|B_i|\le n/2$  for all $i=1,2,3$ and
\item[$(ii)$\hfill]$E(B_i, B_j)=\emptyset$  for all $1\le i < j\le 3$.
\end{list}
\end{lemma}
\begin{proof}
This follows straightforwardly from a greedy type argument. Consider the components of $G$ in any order. Put the components into a set $B_1$ as long as $B_1$ contains at most $n/2$ vertices. Let $C$ be a component whose addition to $B_1$ would increase the size of $B_1$ above $n/2$. Placing $C$ into $B_2$ and all remaining components into $B_3$ concludes the proof of the lemma. \qed
\end{proof}

We denote by $\perm(n)$ the number of perfect matchings  in a complete graph on $n$ vertices. For sake of completeness let us assume that in the case of $n$ odd we count the number of almost perfect matchings. We are interested in the asymptotic behavior of  $\perm(n)$. In fact, we only care on the behavior of the leading terms. With the help of Stirling's formula 
$$
n! = (1+o(1))\cdot \sqrt{2\pi n} \cdot (n/e)^n\enspace ,
$$
one easily obtains that
\begin{equation}\label{eq:perm}
\perm(n) =\frac{n!}{\lfloor n/2 \rfloor ! \cdot 2^{\lfloor n/2 \rfloor }}  = poly(n) \cdot \left( \frac{n}{e} \right)^{n/2},
\end{equation}
where here and in the remainder of this section we use the term $poly(n)$ to suppress  factors that are polynomial in $n$.

For any non-negative real numbers $x, y, z$ such that $x+y+z=1$, we define $H(x, y, z) = -x \log_2 x - y \log_2 y - z \log_2 z$.

\begin{lemma}\label{lemma:boring}
$$\min_{\substack{ 0 \leq x, y, z \leq 1/2\\ x+ y+ z=1}} H(x, y, z) = 1.$$
\end{lemma}
\begin{proof}As $-x\log_2 x$ is concave, $H(x, y, z)$ must attain its minimum in a vertex of the simplex $0 \leq x, y, z \leq 1/2, x+ y+ z=1$, which, up to permutation of the variables, is given by $x=y=\frac{1}{2}$ and $z=0$.\qed
\end{proof}

We are now ready to prove Theorem  \ref{thm:random_matchings_exp_high_prob}, which implies Theorem \ref{thm:random_matchings}.
Let $n=\abs{V}$. In the light of Lemma~\ref{l:alpha:robust}, the theorem follows if we can show that $G=(V, E)=M_1\cup M_2 \cup M_3$ is $\alpha$-robust with exponentially high probability for some sufficiently small $\alpha>0$. By Lemma \ref{lemma:3part}, it suffices to show that, for all $C\subseteq V$ of size at most $\bar{\alpha}n=(1-\alpha) n$ and all partitions $B_1, B_2, B_3$ of $B=V\setminus C$ such that no set $B_i$ contains more than $\abs{B}/2$ vertices, the graph contains an edge that goes between two sets $B_i, B_j$ where $i\neq j$. 

Consider any such partition $B_1, B_2, B_3, C$. For each $i=1, 2, 3$, let $\beta_i = \abs{B_i}/\abs{B}$. Let $M$ be one uniformly chosen perfect matching. 
Let us estimate the probability that $E_M(B_i, B_j)=\emptyset$ for all $i\neq j$. First, condition on the set of edges $M'$ in $M$ with at least one end-point in $C$. These will connect to at most $\alpha n$ vertices in $B$. Let $B', B_1', B_2'$ and $B_3'$ respectively denote the subsets of vertices that remain unmatched. Hence, the remaining edges in the matching $M\setminus M'$ is chosen uniformly from all perfect matchings on $B'$.


We write $\beta_i' = \abs{B_i'}/\abs{B'}$. Clearly, if $\abs{B_i'}$ for some $i=1, 2, 3$ is odd, then $$\mathbb{P}( E_M(B_i, B_j)=\emptyset\;\forall i\neq j | M' )=0.$$ Otherwise, by \eqref{eq:perm}, we get
\begin{align*}
&\mathbb{P}( E_M(B_i, B_j)=\emptyset\;\forall i\neq j | M' ) = \frac{1}{\perm(\abs{B'})} \prod_{i=1}^3 \perm(\abs{B_i'})\\
&\qquad = poly(n) \cdot \left(\frac{e}{\abs{B'}}\right)^{\abs{B'}/2} \prod_{i=1}^3 \left(\frac{ \abs{B_i'}}{e}\right)^{\abs{B_i'}/2}\\
&\qquad = poly(n) \cdot \left(\beta_1'^{\beta_1'}\beta_2'^{\beta_2'}\beta_3'^{\beta_3'}\right)^{
\abs{B'}/2}\\
&\qquad = poly(n) \cdot 2^{-  H(\beta_1', \beta_2', \beta_3')\abs{B'}/2}.
\end{align*}
As $H(x, y, z)$ is uniformly continuous, we know that for any $\varepsilon>0$, there exists an $\alpha_0>0$ such that, for any $0 < \alpha < \alpha_0$, $\abs{(\abs{B}-\abs{B'})/n}$ and $\abs{\beta_i'-\beta_i}$ are sufficiently small that the above expression can be bounded by $2^{-H(\beta_1, \beta_2, \beta_3) \bar{\alpha} n/2 + \varepsilon n}$, where the choice of $\alpha_0$ holds uniformly over all $\beta_1, \beta_2, \beta_3$. As the above bound holds for any matching~$M'$, we get $$\mathbb{P}( E_M(B_i, B_j)=\emptyset\;\forall i\neq j) \leq 2^{-H(\beta_1, \beta_2, \beta_3) \abs{B}/2 + \varepsilon n},$$
for any partition $\{B_1, B_2, B_3, C\}$ as above.

We now do a union bound over all such partitions. For a given set $C$ and given sizes $b_1, b_2, b_3$ of the sets $B_1, B_2, B_3$, the number of choices for these sets is ${\abs{B} \choose b_1, b_2, b_3} = poly(n)\cdot 2^{H(\beta_1, \beta_2, \beta_3) \abs{B} }$ where $B=V\setminus C$. Moreover, by the above calculation, given a partition $B_1, B_2, B_3, C$ as above, the probability that $E(B_i, B_j)=\emptyset$ for a union of three independent uniformly chosen perfect matchings is at most $2^{-3H(\beta_1, \beta_2, \beta_3)  \abs{B}/2 + 3\varepsilon n}.$ This yields
\begin{equation*}
\begin{split}
&\mathbb{P}( \exists B_1, B_2, B_3, C :\,   E(B_i, B_j) = \emptyset\ \forall i\neq j)\\
&\qquad \leq poly(n) \cdot \sum_{\substack{B\subseteq V\\ \abs{B} \ge \bar{\alpha} n}} \sum_{ \substack{ 0\leq b_1, b_2, b_3\leq \abs{B}/2\\ b_1+b_2+b_3 = \abs{B}}} 2^{ -H(\beta_1, \beta_2, \beta_3) \abs{B}/2  + 3\varepsilon n}\\
&\qquad \leq poly(n) \cdot 2^{3\varepsilon n} \sum_{\substack{B\subseteq V\\ \abs{B} \ge \bar{\alpha} n}} \left( \max_{\substack{ 0 \leq \beta_1, \beta_2, \beta_3 \leq 1/2\\ \beta_1+ \beta_2+ \beta_3=1}} 2^{-H(\beta_1, \beta_2, \beta_3)}\right)^{\abs{B} / 2} \\
&\qquad \leq poly(n) \cdot 2^{3\varepsilon n} \sum_{\substack{B\subseteq V\\ \abs{B} \ge \bar{\alpha} n}} 2^{-\abs{B} / 2},
\end{split}
\end{equation*}
where the last line follows by Lemma \ref{lemma:boring}. The remaining sum can be rewritten as $\sum_{k=\lceil \bar{\alpha} n\rceil}^n {n\choose k} 2^{-k / 2}$. Assuming $\bar{\alpha}\ge \frac{1}{2}$, the summand is decreasing. Hence the sum is at most $poly(n)\cdot {n\choose \bar{\alpha} n} 2^{-\bar{\alpha} n/2}$. We conclude that
\begin{equation}\label{eq:ub3m} \mathbb{P}( G\text{ is not $\alpha$-robust})\le poly(n) \cdot 2^{ \left( H(\bar{\alpha})-\bar{\alpha}/2+3\varepsilon\right)n}.
\end{equation}
As $H(\bar{\alpha})-\frac{\bar{\alpha}}{2} = -\frac{1}{2} + H(\alpha) + \frac{\alpha}{2} \rightarrow -\frac{1}{2}$ as $\alpha\rightarrow 0$, we see that choosing $0<\varepsilon< 1/6$ and $\alpha>0$ sufficiently small, the right-hand side of \eqref{eq:ub3m} tends to $0$ exponentially fast in $n$, as desired. \qed

\subsection{Proof of Corollary \ref{cor:3-regular}, \ref{cor:cycle_matching} and \ref{cor:directed_blow_up}}\label{sec:proof_of_corollaries}

\begin{proof}[of Corollary \ref{cor:3-regular}] 
A standard method from random graph theory for the construction of regular graphs is the so-called configuration model introduced by Bollob\'as in~\cite{bollobas1980probabilistic}. Constructing a graph $G$ by taking the union of three independent uniformly chosen random perfect matchings $M_1, M_2, M_3$ is equivalent to constructing a $3$-regular random graph with the configuration model and conditioning that no self-loops appear. Since conditioning that no self-loops appear increases the probability of producing a simple graph (because all graphs with self-loops are not simple), Corollary~\ref{cor:3-regular} thus follows immediately from~\cite{bollobas1980probabilistic}.

For sake of completeness we also give a direct proof. 
By symmetry the probability that a given edge is contained in a random perfect matching is exactly $\frac1{n-1}$. The expected number of edges common to $M_1$ and $M_2$ is thus $\frac{n}{2(n-1)}$. Hence, by Markov's inequality, with probability at least $1-\frac{n}{2(n-1)}\approx\frac12$, $M_1$ and $M_2$ are disjoint. Assuming the two first matchings are disjoint, we uniformly choose one pair of vertices $\{u, v\}$ to form an edge in $M_3$. With probability $1-\frac{2}{n-1}$, this edge is not in $M_1\cup M_2$, and hence shares one end-point with exactly $4$ of the $n$ edges in $M_1\cup M_2$. Adding the remaining $\frac{n}{2}-1$ edges to $M_3$, the expected number of these that are contained in $M_1\cup M_2$ is $\frac{n-4}{n-3} = 1-\frac{1}{n-3}$. Again by Markov's inequality, this means that $M_3$ is disjoint $M_1\cup M_2$ with probability at least $\frac{1}{n-3}$. It follows that $M_1, M_2, M_3$ are pairwise disjoint with probability at least $\Omega\left(\frac{1}{n}\right)$. As Theorem \ref{thm:random_matchings_exp_high_prob} holds with exponentially high probability, for sufficiently large $n$, we can find disjoint matchings $M_1, M_2, M_3$ such that $\alpha n$ cops do not suffice to catch the robber on $G=M_1\cup M_2 \cup M_3$.\qed

\end{proof}

\begin{proof}[of Corollary \ref{cor:cycle_matching}]
Let $M_1,M_2,M_3$ be random perfect matchings chosen  independently and uniformly. We claim that $M_1\cup M_2$ is a Hamiltonian cycle with probability at least $\frac{1}{n-1}$. Therefore,  the graph $G= M_1\cup M_2\cup M_3$ is with probability at least $\frac{1}{n-1}$ a union of a Hamiltonian cycle and a random matching. Since  Theorem \ref{thm:random_matchings_exp_high_prob} holds with probability $1-e^{-\Omega(n)}$,  it holds with high probability that $\alpha n$ cops do not suffice to catch the robber on the graph $G=H\cup M$, where $H$ is an Hamiltonian cycle and $M$ is a random perfect matching.

To see why the claim holds, note that there are $\frac12(n-1)!$ Hamiltonian cycles, each of which can be written as a union $M_1\cup M_2$ of two perfect matchings in two ways, and $(n-1)\cdot(n-3)\cdot\ldots \cdot 1$ perfect matchings. Therefore the probability that the union of two random matchings is Hamiltonian is 
\begin{eqnarray*}
\frac{2\cdot\frac12(n-1)!}{((n-1)\cdot(n-3)\cdot\ldots \cdot 1)^2}=  \frac{1}{n-1} (n-2) \frac{1}{n-3} (n-4)\cdots 2 \frac{1}{1} > \frac{2}{n-1} \ .
\end{eqnarray*}\qed

\end{proof}

\begin{proof}[of Corollary \ref{cor:directed_blow_up}]
Choose an even $n$ large enough. By Theorem \ref{thm:random_matchings}, there exist three perfect matchings $M_1,M_2,M_3$ on $n$ vertices such that $\alpha n$ cops do not suffice to catch the robber on the graph $G'=M_1\cup M_2 \cup M_3 $. We construct the graph $G$ in the following way.  We label the $6n$ vertices by $v_{ij}$ for $1\leq i \leq n$ and $1\leq j \leq 6$. For all $i$, we connect $v_{i1}, \ldots , v_{i6}$ as a directed $6$-cycle. For each edge of one of the matchings, we will connect two of these cycles. More precisely, let $e \in M_k$ and suppose that $e$ connects vertex $i$ and $j$ in $G'$. Then, we add the directed edges $(v_{ik}, v_{j(k+3)})$ and $(v_{jk},v_{i(k+3)})$ to $G$.


We call a cycle $i$ \emph{free} if no cop is on this cycle, and \emph{occupied} otherwise. If the robber enters a free cycle, then he can reach any vertex of the cycle and while doing so, the cops cannot occupy any vertex outside of the cycle.  Consider a situation of the game on graph the $G$ with occupied cycles $F\subset [n]$, where the robber enters a free cycle $i$. This corresponds to the situation on the graph $G'$ with occupied vertices $F$ where the robber enters vertex $i$. On this graph the robber moves according to its winning strategy to a free vertex $j$ with $(i,j)\in M_k$ for some $1\leq k\leq 3$. On the graph $G$ the robber moves first along the cycle $i$ to vertex $v_{ik}$ and then enters the free cycle $j$ via the edge $(v_{ik},v_{j(k+3)})$. Therefore, any winning strategy for the robber on graph $G'$ with $\alpha n$ cops gives a winning strategy for the robber on graph $G$ with $\alpha n$ cops.\qed
\end{proof}

\subsection{Proofs of upper bounds}\label{sec:proofs_upper_bounds}

\begin{proof}[of Proposition \ref{prop:DGn}]
Consider the graph $DG_{n}$ on $2n$ vertices consisting of a Hamilton cycle $u_1, u_2, \ldots, u_n$, $v_1, v_2, \ldots , v_n$ and $n$ additional edges $(u_i,v_i)$ for $1\leq i\leq n$. We call these additional edges \emph{diagonal edges} and the other edges \emph{cycle edges}.
Note that every $DG_n$ graph is $3$-regular. Since the graph cannot be disconnected by removing $2$ vertices, the graph $DG_n$ is also $3$-connected. Furthermore any two pairs of opposite vertices $(u_i, v_i), (u_j,v_j)$ for $i\not= j$ split the graph into two connected components. This structure turns out to be very useful for the cops. By occupying four such vertices the robber is trapped in the connected component he is in.

As in real life, our cops never come alone. Consider three pairs of cops $c_i, c_i'$ for $1\leq i \leq 3$.  In order to catch the robber it suffices that each pair of cops can execute the command \emph{chase robber}. If a pair $i$ of cops is told to chase the robber, $c_i$ and $c_i'$ alternate in flying to the robbers location. This ensures that the robber can never move to its previous location. The first time the robber uses a diagonal edge, this pair of cops  blocks the diagonal, i.e., it stays on the two endpoints of the diagonal edge until it receives the command to chase the robber again. Note that if there is a third cop placed somewhere on the graph, then the robber is forced to use a diagonal edge after less than $2n$ steps. 

In the beginning of the game, when the robber has chosen its starting position, cop $c_3$ flies to the position of the robber. Then, cop pair $c_1, c_1'$ chases the robber. The first time the robber uses a diagonal edge, these cops block the diagonal, and the second pair of cops starts chasing the robber. When the robber uses a diagonal edge again, this pair of cops blocks that diagonal and the third pair of cops $c_3,c_3'$ starts chasing the robber. 

From now on there will  always be two cop pairs blocking two diagonals. Therefore, the robber cannot leave the area between these two diagonals. The remaining pair of cops chases the robber until he moves along a diagonal edge; this diagonal is subsequently blocked by this pair of cops. Note that the robber is now in the component defined by this diagonal and \emph{one of} the two previously blocked diagonals. Correspondingly, one of the two cop pairs is not needed anymore and this pair of cops takes on the chase. The size of the entangled component  the robber is in  decreases by at least $2$  every time the robber uses a diagonal edge. When the component has size $2$ the robber is caught.\qed
\end{proof}

\begin{proof}[of Theorem \ref{thm:upper_bound}] 
Our main approach is the following. First, we identify a set $A$ of size $\lfloor \frac{n}{4}\rfloor$ with the property that  $4$ cops suffices to catch the robber on $G\setminus A$. Then $\lfloor \frac{n}{4}\rfloor+4$ cops can catch the robber using the following strategy. As long as the robber plays on the vertices in $G\setminus A$, we follow the optimal strategy on $G\setminus A$ using at most $4$ cops. Unless this catches the robber, he must at some point move to a vertex in $A$. If that happens, place one cop there permanently and repeat. Eventually, the robber runs out of non-blocked vertices in $A$ to escape to and gets caught. 

In order to identify the vertices that we want to place in $A$, we observe the following. Assume that the strategy on $G\setminus A$ for the cops is such that there will always be cops placed at the robbers last two positions, e.g. two cops take turns to drop on the robber. Then in order to catch the robber in  $G\setminus A$, we can ignore all vertices that have degree one in $G\setminus A$ (as the robber will be caught, if he moves to such a vertex) and replace paths in which all internal vertices have degree two by a single edge  (as two cops suffice to chase the robber along such a path). 

To formalize this, we consider two operations to reduce a graph $G$: $(i)$ remove any degree $1$ vertices, $(ii)$ replace any degree two vertices by an edge. Note that the second operation could create loops or multiple edges, so a reduced graph may not be simple. Moreover, reducing a max degree $3$ graph as far as possible will result in  components that are non-simple $3$-regular or a vertex or a loop. 

We now use these operations to define the set $A$. Pick any vertex of degree 3 and remove it from $G$. This decreases the number of degree $3$ vertices by $4$. Reduce $G$ as far as possible. If the remaining graph has a vertex  with $3$ distinct neighbors, remove it from $G$, and reduce as far as possible. Again, this decreases the number of degree $3$ vertices by $4$. Repeat until no vertex in $G$ has more than $2$ distinct neighbors. Let $A$ be the set of removed vertices. Then $|A|\le \lfloor \frac{n}{4}\rfloor$.

It remains to show that $4$ cops can catch the robber on $G\setminus A$. By construction, any connected component of $G\setminus A$ can be reduced to either a loop, a vertex, or a (non-simple) graph where all vertices have degree $3$ but at most two neighbors. It is easy to see that the only way to satisfy these properties is for each vertex to either be incident to one loop and one single edge, one single edge and one double edge, or one triple edge. Thus, the only possible graphs are $(a)$ an even cycle where every second edge is double, $(b)$ a path where every second edge is double and where the end-points have loops attached, or $(c)$ two vertices connected by a triple edge. We leave it as an exercise to see that the robber can be caught on any such graphs using at most $4$ cops. Note that, due to $(ii)$, it might be possible for the robber to have the same position at times $t$ and $t+2$ by following a pair of edges between the same two vertices, but this can be prevented using one of the cops. \qed

\end{proof}

\section{Conclusion} \label{sec:Conclusions}
We have shown that there are many  graphs with maximum degree three for which the entanglement number is of linear size, in the directed cases as well as in the undirected case.
This shows that the freedom of the cops of being able to \emph{fly} to any vertex is not helpful when they are only allowed to fly to the current position of the robber. In other words, they should not only {\em follow} the robber, but they should think ahead of where the robber might want to go, as the title of our paper indicates.

All our examples were found by taking the the union of three random matchings. 
We have shown that there exists an $\alpha$ such that with high probability, the robber can run away from $\alpha n$ cops. We also showed that $\lfloor\frac14n\rfloor + 4$ cops do suffice.
We leave it as an open problem to determine  the exact value of $\alpha$.

\subsubsection*{Acknowledgments}
We thank Malte Milatz for bringing this problem to our attention.



\begin{footnotesize}
\bibliographystyle{plain}
\bibliography{sources}
\end{footnotesize}

\end{document}